\documentclass[a4paper,UKenglish]{lipics-v2018}

\usepackage{microtype}
\usepackage{todonotes}
\usepackage{xspace}

\newcommand{\lca}{\mbox{\it lca}}

\newcommand{\parent}{\mathit{p}}

\nolinenumbers 

\title{A Novel Algorithm for the All-Best-Swap-Edge Problem on Tree Spanners}

\titlerunning{A Novel Algorithm for the ABSE Problem on Tree Spanners}

\author{Davide Bil\`o}{Department of Humanities and Social Sciences, University of Sassari, Italy}{davidebilo@uniss.it}{https://orcid.org/0000-0003-3169-4300}{}

\author{Kleitos Papadopoulos}{InSPIRE, Agamemnonos 20, Nicosia, 1041, Cyprus}{kleitospa@gmail.com}{https://orcid.org/0000-0002-7086-0335}{}

\authorrunning{D. Bil\`o and K. Papadopoulos}

\Copyright{Davide Bil\`o and Kleitos Papadopoulos}

\subjclass{\ccsdesc[100]{Theory of computation~Graph algorithms analysis}}

\keywords{Transient edge failure, best swap edges, tree spanner.}

\category{}

\relatedversion{}

\supplement{}

\funding{}

\acknowledgements{}

\begin{document}

\maketitle

\begin{abstract}
Given a 2-edge connected, unweighted, and undirected graph $G$ with $n$ vertices and $m$ edges, 
a $\sigma$-tree spanner is a spanning tree $T$ of $G$ in which the ratio between the distance in $T$ of any pair of vertices and the corresponding distance in $G$ is upper bounded by $\sigma$. The minimum value of $\sigma$ for which $T$ is a $\sigma$-tree spanner of $G$ is also called the {\em stretch factor} of $T$. We address the fault-tolerant scenario in which each edge $e$ of a given tree spanner may temporarily fail and has to be replaced by a {\em best swap edge}, i.e. an edge that reconnects $T-e$ at a minimum stretch factor. More precisely, we design an $O(n^2)$ time and space algorithm that computes a best swap edge of every tree edge. Previously, an $O(n^2 \log^4 n)$ time and $O(n^2+m\log^2n)$ space algorithm was known for edge-weighted graphs [Bil\`o et al., ISAAC 2017]. Even if our improvements on both the time and space complexities are of a polylogarithmic factor, we stress the fact that the design of a $o(n^2)$ time and space algorithm would be considered a breakthrough.
\end{abstract}

\section{Introduction}

Given a 2-edge connected, unweighted, and undirected graph $G$ with $n$ vertices and $m$ edges, 
a $\sigma$-tree spanner is a spanning tree $T$ of $G$ in which the ratio between the distance in $T$ of any pair of vertices and the corresponding distance in $G$ is upper bounded by $\sigma$. The minimum value of $\sigma$ for which $T$ is a $\sigma$-tree spanner of $G$ is also called the {\em stretch factor} of $T$. The stretch factor of a tree spanner is a measure of how the all-to-all distances degrade w.r.t. the underlying communication graph if we want to sparsify it. Therefore, tree spanners find several applications in the network design problem area as well as in the area of distributed algorithms (see also~\cite{IIOY03,Pro00} for some additional practical motivations).

Unfortunately, tree-based network infrastructures are highly sensitive to even a single transient link failure, since this always results in a network disconnection. Furthermore, when these events occur, the computational costs for rearranging the network flow of information from scratch (i.e., recomputing a new tree spanner with small stretch factor, reconfiguring the routing tables, etc.) can be extremely high. Therefore, in such cases it is enough to promptly reestablish the network connectivity by the addition of a {\em swap} edge, i.e. a link that temporarily substitutes the failed edge.

In this paper we address the fault-tolerant scenario in which each edge $e$ of a given tree spanner may undergo a transient failure and has to be replaced by a {\em best swap edge}, i.e. an edge that reconnects $T-e$ at a minimum stretch factor. More precisely, we design an $O(n^2)$ time and space algorithm that computes all the best swap edges (ABSE for short) in unweighted graphs, that is a best swap edge for every edge of $T$. Previously, an $O(n^2 \log^4 n)$ time and $O(n^2+m\log^2n)$ space algorithm was known for edge-weighted graphs. Even though the overall improvements in both the time and space complexities are of a polylogarithmic factor, we stress the fact that designing an $o(n^2)$ time and space algorithm would be considered a breakthrough in this field (see~\cite{DBLP:conf/isaac/BiloCG0P17}). Furthermore, the approach proposed in this paper uses only one technique provided in~\cite{DBLP:conf/sirocco/BiloCG0P15}; all the remaining ideas are totally new and are at the core of the design of both a time and space efficient algorithm. Our algorithm is also easy to implement and makes use of very simple data structures.  

\subsection{Related work}

The ABSE problem on tree spanners has been introduced by Das et al. in~\cite{DBLP:journals/jgaa/DasGW10}, where the authors designed two algorithms for both the weighted and the unweighted case, running in $O(m^2 \log n)$ and $O(n^3)$ time, respectively, and using $O(m)$ and $O(n^2)$ space, respectively. Subsequently, Bilò et at.~\cite{DBLP:conf/sirocco/BiloCG0P15} improved both results by providing two efficient linear-space solutions for both the weighted and the unweighted case, running in $O(m^2 \log \alpha(m,n))$ and $O(m n  \log n)$ time, respectively. Recently, in~\cite{DBLP:conf/isaac/BiloCG0P17} the authors designed a very clever recursive algorithm that uses centroid-decomposition techniques and lower envelope data structures to solve the ABSE problem on tree spanners in $O(n^2 \log^4 n)$ time and $O(n^2+m\log^2 n)$ space. 
Table \ref{table} summarizes the state of the art for the ABSE problem on tree spanners.

\begin{table}[ht]
\caption{The state of the art for the ABSE problem on tree spanners. The naive algorithm works as follows: for each edge $e$ of the tree spanner $T$ (that are $O(n)$), we look at all the possible swap edges (that are $O(m)$) and, for each swap edge $f$, we compute the stretch factor of $T$ where $e$ is swapped with $f$ (this requires $O(n^2)$). We observe that the naive algorithm needs to store the all-to-all (post-failure) distances in $G-e$.} 
\centering 
\begin{tabular}{|l|c|c|c|c|} 
\hline
\multirow{2}{*}{\bf Algorithm}						 & \multicolumn{2}{c|}{\bf weighted graphs} 				& \multicolumn{2}{c}{\bf unweighted graphs}\vline \\
\cline{2-5}
												& {\bf time}					& {\bf space}			& {\bf time}	& {\bf space}\\
\hline
naive											& $\Theta(n^3m)$ 					& $\Theta(n^2)$ 				& $\Theta(n^3m)$		& $\Theta(n^2)$ \\
\hline
Das et al.~\cite{DBLP:journals/jgaa/DasGW10}& $O(m^2 \log n)$ 				& $O(m)$				& $O(n^3)$		& $O(n^2)$\\
\hline
Bilò et al.~\cite{DBLP:conf/sirocco/BiloCG0P15}	& $O(m^2 \log \alpha(m,n))$ 	& $O(m)$				& $O(mn \log n)$	& $O(m)$\\
\hline
Bilò et al.~\cite{DBLP:conf/isaac/BiloCG0P17}	& $O(n^2 \log^4 n)$ 			& $O(n^2+m\log^2n)$		& $O(n^2 \log^4 n)$ 			& $O(n^2+m\log^2n)$\\
\hline
this paper		 								& -							 	& -						& $O(n^2)$					& $O(n^2)$\\
\hline
\end{tabular}
\label{table}
\end{table}

\subsection{Other related work on ABSE}

The ABSE problems in spanning trees have received a lot of attention from the algorithmic community.
The most famous and first studied ABSE problem was on \emph{minimum spanning trees}, where the quality of a swap edge is measured w.r.t. the overall cost of the resulting tree (i.e., sum of the edge weights). This problem, a.k.a. \emph{sensitivity analysis} problem on minimum spanning trees, can be solved in $O(m\log\alpha(m,n))$ time~\cite{Pet05}, where $\alpha$ denotes the inverse of the Ackermann function. 
In the {\em minimum diameter spanning tree} a quality of a swap edge is measured w.r.t. the \emph{diameter} of the swap tree~\cite{IR98,NPW01}. Here the ABSE problem can also be solved in $O(m \log \alpha(m,n))$ time~\cite{BGP15}. In the {\em minimum routing-cost spanning tree}, the best swap minimizes the overall sum of the all-to-all distances of the swap tree~\cite{WHC08}. The fastest algorithm for solving the ABSE problem in this case has a running time of $O\left(m 2^{O(\alpha(n,n))}\log^2 n\right)$~\cite{BGP14}. Concerning the {\em single-source shortest-path tree}, several criteria for measuring the quality of a swap edge have been considered. The most important ones are:
\begin{itemize}
\item the maximum or the average distance from the root; here the corresponding ABSE problems can be solved in  $O(m \log \alpha(m,n))$ time (see~\cite{BGP15}) and $O(m \, \alpha(n,n) \log^2 n)$ time (see~\cite{DP07}), respectively;
\item the maximum and the average stretch factor from the root for which the corresponding ABSE problems have been solved in $O(m n  +n^2 \log n)$ and $O(m n \log \alpha(m,n))$ time, respectively~\cite{DBLP:conf/sirocco/BiloCG0P17}.
\end{itemize}
Finally, the ABSE problems have also been studied in a distributed setting~\cite{FEPPS04, FEPPS06, FPPSW08}.

\section{Preliminary definitions}
Let $G = (V(G), E(G))$ be a $2$-edge-connected, unweighted, and undirected graph of $n$ vertices and $m$ edges, respectively, and let $T$ be a spanning tree of $G$. Given an edge $e \in E(G)$, we denote by $G-e=(V(G),E(G)\setminus\{e\})$ the graph obtained after the removal of $e$ from $G$.
Given an edge $e \in E(T)$, let $S(e)$ denote the set of all the \emph{swap edges} of $e$, i.e., all edges in $E(G) \setminus \{ e \}$ whose endpoints belong to two different connected components of $T-e$.
For any $e \in E(T)$ and $f \in S(e)$, let $T_{e/f}$ denote the \textit{swap tree} obtained from $T$ by replacing $e$ with $f$. Given two vertices $x,y \in V(G)$, we denote by $d_G(x,y)$ the \emph{distance} between $x$ and $y$ in $G$, i.e., the number of edges contained in a shortest path in $G$ between $x$ and $y$. We define the \textit{stretch factor} $\sigma_G(T)$ of $T$ w.r.t. $G$ as 
\begin{equation*}
\sigma_G(T) = \max_{x,y \in V(G)}  \frac{d_T(x,y)}{d_G(x,y)}.
\end{equation*}

\begin{definition}[Best Swap Edge]
	Let $e \in E(T)$. An edge $f^* \in S(e)$ is a \textit{best swap edge} for $e$ if $f^* \in \arg\min_{f \in S(e)} \sigma_{G-e}\big(T_{e/f}\big)$.
\end{definition}
For a rooted tree $T$ and two vertices $u$ and $v$ of $T$, we denote by $\mathcal{A}(v)$ the set of all the proper ancestors of $v$ in $T$, we denote by $p(v)$ the parent of $v$ in $T$, and we denote by $\lca(u,v)$ the least common ancestor of $u$ and $v$ in $T$.
\section{The algorithm}

In this section we design an $O(n^2)$ time and space algorithm that computes a best swap edge for every edge of $T$. Let $r$ be an arbitrarily chosen vertex of $T$. For the rest of the paper, we assume that $T$ is rooted at $r$. The algorithm works as follows. First, for every vertex $x$ of $T$, the algorithm computes the set $E(x):=\big\{(x,y) \in E(G)\setminus E(T) \mid x \not \in {\mathcal A}(y)\big\}$ of non-tree edges of the form $(x,y)$, where $x$ is not an ancestor of $y$ in $T$ (see Figure \ref{fig:edge_subdivision}). Observe that some sets $E(x)$ may be empty. Observe also that each edge $(x,y)$ such that $x \not\in {\mathcal A}(y)$ and $y \not\in {\mathcal A}(x)$ is contained in both $E(x)$ and $E(y)$. The precomputation of the all the sets $E(x)$ requires linear time if we use a data structure that can compute the least common ancestor of any 2 given vertices in constant time~\cite{DBLP:conf/stoc/GabowT83}.
\begin{figure}[t]
	\centering
	\includegraphics[scale=1.1]{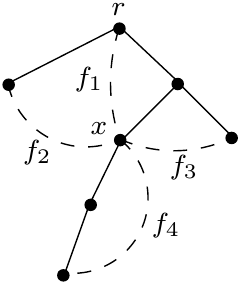}
	\caption{An example showing how the set $E(x)$ is defined. Tree edges are solid, while swap edges are dashed. In this example $E(x)=\{f_1,f_2,f_3\}$.}
	\label{fig:edge_subdivision}
\end{figure}

The algorithm visits the edges of $T$ in postorder and, for each edge $e \in E(T)$, it computes a corresponding best swap edge in $O(n)$ time. For the rest of the paper, unless stated otherwise, let $e=(p(v),v)$ be a fixed tree edge and let $X$ be the set of vertices contained in the subtree of $T$ rooted at $v$. The algorithm computes a best swap edge $f^*$ of $e$ as follows. First, for every $x \in X$, the algorithm computes a \emph{candidate}  best swap edge $f_{x}$ of $e$ that is chosen among the edges of $F(x,e):=E(x) \cap S(e)$.\footnote{With a little abuse of notation, if $F(x,e)=\emptyset$, then $f_x=\perp$ and $\sigma_{G-e}(T_{e/f_x})=+\infty$.} More precisely, 
\[
f_x \in \arg\min_{f \in F(x,e)}\sigma_{G-e}(T_{e/f}).
\]
The best swap edge $f^*$ is then selected among the computed candidate best swap edges. More precisely, 
\begin{equation}\label{eq:best_swap_edge}
f^* \in \arg\min_{f_x \text{ s.t. }x \in X}\sigma_{G-e}\big(T_{e/f_x}\big).
\end{equation}
We can prove the following lemma.
\begin{lemma}\label{lm:best_swap_edge_computation}
The edge $f^*$ computed as in~(\ref{eq:best_swap_edge}) is a best swap edge of $e$.
\end{lemma}
\begin{proof}
Let $x \in X$ and let $(x,y) \in S(e)$ be any swap edge of $e$ incident to $x$. Since $x \not\in {\mathcal A}(y)$, we have that $(x,y) \in E(x)$. Therefore, $(x,y) \in F(x,e)$. As a consequence, $S(e) = \bigcup_{x \in X} F(x,e)$. Hence
\begin{equation*}
\sigma_{G-e}\big(T_{e/f^*}\big) = \min_{f_x \text{ s.t. } x \in X}\sigma_{G-e}\big(T_{e/f_x}\big)  = \min_{x \in X}\min_{f \in F(x,e)} \sigma_{G-e}\big(T_{e/f}\big)=\min_{f \in S(e)} \sigma_{G-e}\big(T_{e/f}\big).
\end{equation*}
The claim follows.
\end{proof}

\subsection{How to compute the candidate best swap edges}\label{subsec:how_to_compute_the_candidate_best_swap_edge}

As already proved in Lemma 3 of~\cite{DBLP:conf/sirocco/BiloCG0P15}, the candidate best swap edge $f_x$ can be computed via a reduction to the \emph{subset minimum eccentricity problem on trees}. We revise the reduction in the following. 
In the \emph{subset minimum eccentricity problem on trees}, we are given a tree ${\mathcal T}$, with a cost $c(y)$ associated with each vertex $y$, and a subset $Y\subseteq V({\mathcal T})$, and we are asked to find a vertex in $Y$ of \emph{minimum eccentricity}, i.e., a vertex $y^* \in Y$ such that
\begin{equation*}
y^* \in \arg\min_{y \in Y}\max_{y' \in V({\mathcal T})}\big(d_{{\mathcal T}}(y,y')+c(y')\big).
\end{equation*}

The reduction from the problem of computing the candidate best swap edge $f_x$ to the subset minimum eccentricity problem on trees is as follows. The input tree corresponds to $T$, the cost associated with each vertex $y$ is $c_x(y) :=\max_{x' \in X, (x',y) \in S(e)}d_{T}(x',x)$,\footnote{If $S(e)$ contains no edge incident to $y$, then $c_x(y)=-\infty$.} and the subset of vertices from which we have to choose the one with minimum eccentricity is $Y(x,e):=\big\{y \mid (x,y) \in F(x,e)\big\}$. As the following lemma shows, the problem can be solved by computing:
\begin{itemize}
\item the endvertices of a \emph{diametral path} of $T$, i.e., two (not necessarily distinct) vertices $a_x, b_x \in V(T)$ such that
\[
\{a_x,b_x\} \in \arg\max_{\{a,b\}, a,b \in V(T) }\big (c_x(a)+d_{T}(a,b)+c_x(b)\big);
\]
\item a \emph{center} of $T$, i.e., a vertex $\gamma_x \in V(T)$ such that
\[
\gamma_x \in \arg\min_{\gamma \in V(T)}\max_{y \in V(T)}\big(d_{T}(\gamma,y)+c_x(y)\big).
\]
\end{itemize}

\begin{lemma}[Bilò et al.\cite{DBLP:conf/sirocco/BiloCG0P15}, Lemma 6 and Lemma 7]\label{lm:sirocco}
Let $\gamma_x$ be a center of $T$ and let $a_x$ and $b_x$ be the two endvertices of a diametral path $P$ in $T$. Then $\gamma_x$ is also a center of $P$. Furthermore, if $y_x \in Y(x,e)$ is the vertex closest to the center $\gamma_x$, i.e., $y_x \in \arg\min_{y \in Y(x,e)} d_{T}(y,\gamma_x)$, then $f_x:=(x,y_x)$ is a candidate best swap edge of $e$ and $\sigma_{G-e}\big(T_{e/f_x}\big) = 1+\max\big\{d_T(y_x,a_x)+c_x(a_x),d_T(y_x,b_x)+c_x(b_x)\big\}$.
\end{lemma}

In what follows we show how all the vertices $y_x$ and all the values $\sigma_{G-e}\big(T_{e/f_x}\big)$, for every $x \in X$, can be computed in $O(n)$ time and space. More precisely, the algorithm first computes the endvertices $a_x$ and $b_x$, for every $x \in X$, in $O(n)$ time and space. Thanks to Lemma~\ref{lm:sirocco}, once both $a_x$ and $b_x$ are known, and since all tree edges have length equal to 1, we can compute $\gamma_x$ in constant time using a constant number of least common ancestor and {\em level ancestor} queries~\cite{DBLP:journals/tcs/BenderF04}.\footnote{Indeed, by computing the least common ancestor between $a_x$ and $b_x$, say $\bar x$, we know whether $\gamma_x$ is along either the $\bar x$ to $a_x$ path or the $\bar x$ to $b_x$ path. If $\gamma_x$ is an ancestor of $a_x$, then its distance from $a_x$ is equal to $\big\lceil (c_x(a_x)+d_T(a_x,b_x)+c_x(b_x))/2 \big\rceil - c_x(a_x)$. If $\gamma_x$ is an ancestor of $b_x$, then its distance from $b_x$ is equal to $\big\lceil (c_x(a_x)+d_T(a_x,b_x)+c_x(b_x))/2 \big\rceil - c_x(b_x)$.}
Finally, for each $x \in X$, we show how to compute the vertex $y_x$ that is closest to $\gamma_x$ in constant time using {\em range-minimum-query} data structures~\cite{DBLP:journals/tcs/BenderF04, DBLP:conf/stoc/GabowT83}.

\subsubsection{How to compute the endvertices of the diametral paths}

To compute $a_x$ and $b_x$, we make use of the following key lemma.

\begin{lemma}[Merge diameter lemma]\label{lm:merge_diameter_lemma}
Let $T$ be a tree, with a cost $c(y)$ associated with each $y \in V(T)$. Let $c_1,\dots,c_\ell$ be $\ell$ (vertex-cost) functions and let $k_1,\dots,k_\ell$ be $\ell$ constants such that, for every vertex $y \in V(T)$, $c(y) = \max_{i=1,\dots,\ell}\big(c_i(y)+k_i\big)$. For every $i=1,\dots,\ell$, let $a_i,b_i$ be the two endvertices of a diametral path of $T$ w.r.t. the cost function $c_i$. Then, there are two indices $i,j=1,\dots, \ell$ ($i$ may also be equal to $j$) and two vertices $a \in \{a_i,b_i\}$ and $b \in \{a_j,b_j\}$ such that:
\begin{enumerate}
\item $a$ and $b$ are the two endvertices of a diametral path of $T$ w.r.t. cost function $c$;
\item $c(a)=c_i(a)+k_i$;
\item $c(b)=c_j(b)+k_j$.
\end{enumerate}
Furthermore, if $a_i$, $b_i$, and their corresponding costs $c_i(a_i)$ and $c_i(b_i)$ are known for every $i=1,\dots,\ell$, then the vertices $a$ and $b$ can be computed in $O(\ell)$ time and space.
\end{lemma}
\begin{proof}
Let $a,b$ be the two endvertices of a diametral path in $T$ w.r.t. the cost function $c$. For some $i,j=1,\dots,\ell$, we have that $c(a)=c_{i}(a)+k_i$ as well as $c(b)=c_j(b)+k_j$ ($i$ may also be equal to $j$). Let $P_i$ (resp., $P_j$) be the path in $T$ between $a_i$ (resp., $a_j$) and $b_i$ (resp., $b_j$). Let $t$ be the first vertex of the path in $T$ from $a$ to $a_i$ that is also in $P_i$, where we assume that the path is traversed in the direction from $a$ to $a_i$. Similarly, let $t'$ be the first vertex of the path in $T$ from $b$ to $b_j$ that is also in $P_j$, where we assume that the path is traversed in the direction from $b$ to $b_j$. We claim that there are $\bar a \in \{a_i,b_i\}$ and $\hat b \in \{a_j,b_j\}$ such that 
\begin{equation}\label{eq:merge_diameter_lemma_ter}
d_T(\bar a, \hat b)=d_T(\bar a, t)+d_T(t,t')+d_T(t',\hat b).
\end{equation}
Indeed, we observe that at least one of the two paths in $T$ from $a_i$ to $t'$ and from $b_i$ to $t'$ passes through $t$. W.l.o.g., we assume that the path in $T$ from $a_i$ to $t'$ passes through $t$ (see Figure~\ref{fig:merge_diameter_lemma}). Similarly, at least one of the two paths in $T$ from $a_i$ to $a_j$ and from $a_i$ to $b_j$ passes through $t'$. As a consequence, such a path also passes through $t$. Therefore $\bar a = a_i$ and $\hat b \in \{a_j,b_j\}$ (see Figure~\ref{fig:merge_diameter_lemma}).
\begin{figure}[t]
	\centering
	\includegraphics[scale=1.0]{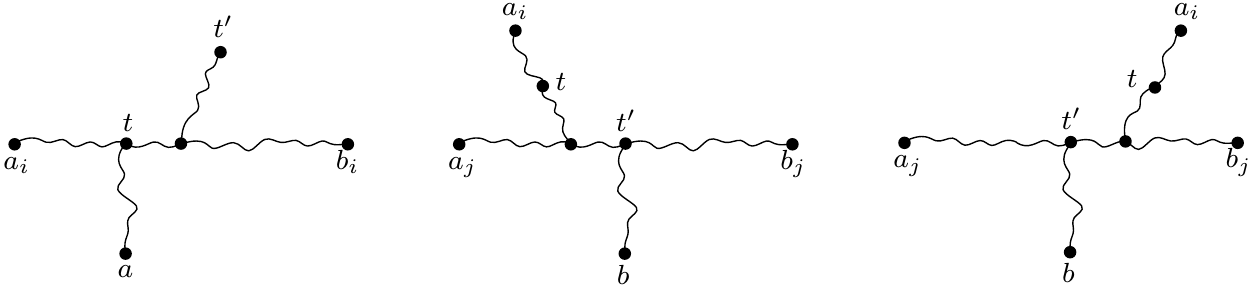}
	\caption{On the left side the path from $a_i$ to $t'$ passes through $t$. In the middle, the path from $a_i$ to $b_j$ passes through $t'$, and thus to $t$. On the right side, the path from $a_i$ to $a_j$ passes through $t'$, and thus to $t$.}
	\label{fig:merge_diameter_lemma}
\end{figure}

Let $\bar b \in \{a_i,b_i\}$, with $\bar b \neq \bar a$, and $\hat a \in \{a_j,b_j\}$, with $\hat a \neq \hat b$. Since $\bar a$ and $\bar b$ are the endvertices of a diametral path in $T$ w.r.t. the cost function $c_i$, we have that
\begin{align*}
c_i(a)+d_T(a,t)+d_T(t,\bar b) + c_i(\bar b) 	& = c_i(a)+d_T(a,\bar b) + c_i(\bar b)\\
												& \leq c_i(\bar a)+d_T(\bar a,\bar b) + c_i(\bar b)\\
												& = c_i(\bar a)+d_T(\bar a,t)+d_T(t,\bar b)+c_i(\bar b),
\end{align*}
from which we derive
\begin{equation}\label{eq:merge_diameter_lemma}
c_i(a)+d_T(a,t)\leq c_i(\bar a)+d_T(\bar a,t).
\end{equation}
Similarly, since $\hat a$ and $\hat b$ are the endvertices of a diametral path in $T$ w.r.t. the  cost function $c_j$, we have that
\begin{align*}
c_j(\hat a)+d_T(\hat a,t')+d_T(t', b) + c_j(b)	& = c_j(\hat a)+d_T(\hat a, b) + c_j(b)\\
												& \leq c_j(\hat a)+d_T(\hat a,\hat b) + c_j(\hat b)\\
												& = c_j(\hat a)+d_T(\hat a,t')+d_T(t',\hat b)+c_j(\hat b),
\end{align*}
from which we derive
\begin{equation}\label{eq:merge_diameter_lemma_bis}
d_T(t',b)+c_j(b)\leq d_T(t',\hat b)+c_j(\hat b).
\end{equation}
Using Inequality (\ref{eq:merge_diameter_lemma}) and Inequality (\ref{eq:merge_diameter_lemma_bis}), together with Equality (\ref{eq:merge_diameter_lemma_ter}), we obtain
\begin{align*}
c(a)+d_T(a,b)+c(b)	& \leq c(a)+d_T(a,t) + d_T(t,t') + d_T(t',b) + c(b)\\
					& = c_i(a)+k_i + d_T(a,t) + d_T(t,t') + d_T(t',b) + c_j(b)+k_j\\
					& \leq c_i(\bar a) + k_i + d_T(\bar a, t) + d_T(t,t') + d_T(t',\hat b) + c_j(\hat b) + k_j\\
					& = c_i(\bar a) + k_i + d_T(\bar a, \hat b) + c_j(\hat b) + k_j\\
					& \leq c(\bar a) + d_T(\bar a, \hat b) + c(\hat b).
\end{align*}
Since $a$ and $b$ are the two endvertices of a diametral path in $T$ w.r.t. cost function $c$, the above inequality is satisfied at equality. As a consequence, $a=\bar a$ and $b=\hat b$ satisfy all the three conditions of the lemma statement.

We complete the proof by showing that $a$ and $b$ can be computed in $O(\ell)$ time using dynamic programming. For every $i=1,\dots,\ell$, we compute the endvertices $\alpha_i$ and $\beta_i$ of a diametral path in $T$ w.r.t. the cost function $\psi_i:=\max_{1\leq j\leq i}\big(c_j(y)+k_j\big)$, together with their corresponding costs. Clearly, for $i=1$, $\alpha_1=a_1, \beta_1=b_1, \psi_1(\alpha_1)=c_1(a_1)+k_1$, and $\psi_1(\beta_1)=c_1(b_1)+k_1$. Moreover, for every $i\geq 2$, we can compute $\alpha_i$ and $\beta_i$, together with $\psi_i(\alpha_i)$ and $\psi_i(\beta_i)$, in constant time and space from $\alpha_{i-1}$ and $\beta_{i-1}$, where $\psi_i(x)=\psi_{i-1}(x)$ for $x \in \{\alpha_{i-1},\beta_{i-1}\}$, and from $a_i$ and $b_i$, where $\psi_i(x)=c_i(x)+k_i$ for $x \in \{a_i,b_i\}$. Therefore, $\alpha_\ell$ and $\beta_\ell$, together with $\psi_\ell(\alpha_\ell)$ and $\psi_\ell(\beta_\ell)$, can be computed in $O(\ell)$ time and space. The claim follows by observing that $a=\alpha_{\ell}$ and $b=\beta_{\ell}$.

\end{proof}

Lemma~\ref{lm:merge_diameter_lemma} is extensively used by our algorithm to precompute some useful information.
For every $x \in V(T)$ and every $z \in \mathcal{A}(x)$, the algorithm precomputes $Y(x|z)=\big\{y \mid (x,y) \in E(x) \text{ and } \lca(x,y)=z\big\}$. Next, the algorithm computes the two endvertices $a_{x|z}$ and $b_{x|z}$ of a diametral path of $T$, together with their associated costs, w.r.t. the following cost function:
\begin{equation*}\label{eq:atomic_cost_functions}
c_{x|z}(y):=\begin{cases}
				0		&	\text{if $y \in Y(x|z)$;}\\
				-\infty &	\text{otherwise.}
			\end{cases}
\end{equation*}
\begin{lemma}\label{lm:endvertices_vertex_vertex}
For every $x \in V(T)$ and every $z \in \mathcal{A}(x)$, all the vertices $a_{x|z}, b_{x|z}$ and their corresponding costs w.r.t. $c_{x|z}$ can be computed in $O(n^2)$ time and space.
\end{lemma}
\begin{proof}
We show that, for any $x \in V(T)$ and any $z \in \mathcal{A}(x)$, the vertices $a_{x|z}$ and $b_{x|z}$ can be computed in $O\big(1+|Y(x|z)|\big)$ time and space. The claim would follow immediately since 
\begin{equation*}
\sum_{x \in V(T)}\sum_{z \in \mathcal{A}(x)}O\Big(1+\big|Y(x|z)\big|\Big)=\sum_{x \in V(T)}O\left(\sum_{z \in \mathcal{A}(x)}\big(1+\big|Y(x|z)\big|\big)\right)=\sum_{x \in V(T)}O(n)=O(n^2).
\end{equation*}
Let $x \in V(T)$ and $z \in \mathcal{A}(x)$ be fixed, and let $\ell=\big|Y(x|z)\big|$. Let $y_1,\dots,y_\ell$ be the $\ell$ vertices of $Y(x|z)$ and, finally, for every $i=1,\dots,\ell$, let
\begin{equation*}
c_{i}(y):=	\begin{cases}
					0		&	\text{if $y=y_i$;}\\
					-\infty	&	\text{otherwise.}
				\end{cases}
\end{equation*}
We have that $a_i=b_i=y_i$ are the two endvertices of the unique diametral path in $T$ w.r.t. cost function $c_i$. Moreover, for every $y \in V(T)$, we have that 
$c_{x|z}(y)=\max_{i=1,\dots,\ell}c_i(y)$. Therefore, using Lemma~\ref{lm:merge_diameter_lemma}, we can compute $a_{x|z}, b_{x,z}$, and their corresponding costs w.r.t. $c_{x|z}$, in $O\big(1+|Y(x|z)|\big)$ time and space.

\end{proof}
Let $c_{x,e}$ be a cost function that, for every $y \in V(T)$, is defined as follows
$c_{x,e}(y):=\max_{z \in \mathcal{A}(v)}c_{x|z}(y)$.
The algorithm also precomputes the two endvertices $a_{x,e}$ and $b_{x,e}$ of a diametral path of $T$ w.r.t. the cost function $c_{x,e}$, together with the corresponding values $c_{x,e}(a_{x,e})$ and $c_{x,e}(b_{x,e})$. The following lemma holds.
\begin{lemma}\label{lm:endvertices_vertex_edge}
For every $x \in V(T)$ and every edge $e$ in the path in $T$ between $r$ and $x$, all the vertices $a_{x,e},b_{x,e}$ and their corresponding costs w.r.t. $c_{x,e}$ can be computed in $O(n^2)$ time and space.
\end{lemma}
\begin{proof}
Let $x \in V(T)$ and let $e$ be an edge of the path between $r$ and $x$ in $T$. We show that $a_{x,e},b_{x,e}$ (and their corresponding costs w.r.t. $c_{x,e}$) can be computed in constant time and space. The claim then follows since $V(T),E(T)=O(n)$. We divide the proof into two cases. 

The first case occurs when $e$ is incident to $r$. Clearly, for every $y \in V(T)$, $c_{x,e}(y)=c_{x|r}(y)$. As a consequence, $a_{x,e}=a_{x|r}$ and $b_{x,e}=b_{x|r}$.

The second case occurs when $e=\big(u=\parent(v),v\big)$, with $u\neq r$. Let $e'=\big(\parent(u),u\big)$. Then, for every $y \in V(T)$, $c_{x,e}(y)=\max\big\{c_{x,e'}(y),c_{x|u}(y)\big\}$.

Therefore, using Lemma \ref{lm:merge_diameter_lemma}, all the vertices $a_{x,e},b_{x,e}$ (and their corresponding costs w.r.t. $c_{x,e}$), with $x \in V(T)$ and $e$ in the path in $T$ between $r$ and $x$, can be computed in $O(n^2)$ time and space via a preorder visit of the tree edges. 
\end{proof}

In the following we show how to compute $a_x$ and $b_x$, for every $x \in X$, in $O(n)$ time and space. First, for every $x \in X$, we consider a subdivision of $X$ into three sets $Z(x,1), Z(x,2)$, and $Z(x,3)$ (see Figure \ref{fig:zeta_seta}) such that:
\begin{itemize}
\item $Z(x,1)$ is the set of all the descendants of $x$ in $T$ ($x$ included);
\item $Z(x,2)$ is the union of all the sets $Z(s,1)$, for every sibling $s$ of $x$;
\item $Z(x,3)=X\setminus \big(Z(x,1)\cup Z(x,2)\big)$.
\end{itemize}
\begin{figure}[t]
	\centering
	\includegraphics[scale=1.1]{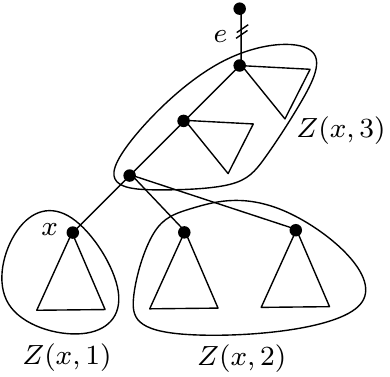}
	\caption{An example showing how the set $Z(x,1), Z(x,2)$ e $Z(x,3)$ are defined. Tree edges are solid, while triangles are subtrees.}
	\label{fig:zeta_seta}
\end{figure}

Each set $Z(x,i)$ is associated with a cost functions $c_{x,i}$ that, for every $y \in V(T)$, is defined as follows:
\begin{equation*}
c_{x,i}(y):=	\max_{x' \in Z(x,i),(x',y) \in S(e)}d_{T}(x',x).
\end{equation*}
Let $a_{x,i}$ and $b_{x,i}$ be the two endvertices of a diametral path in $T$ w.r.t. cost function $c_{x,i}$. The algorithm computes all the vertices $a_{x,i},b_{x,i}$ and their corresponding values w.r.t. cost function $c_{x,i}$, for every $x \in X$ and every $i=1,2,3$.
\begin{lemma}\label{lm:endvertices_vertex_index}
For every $x \in X$ and every $i\in \{1,2,3\}$, all the vertices $a_{x,i},b_{x,i}$ and their corresponding costs w.r.t. $c_{x,i}$ can be computed in $O(n)$ time and space.
\end{lemma}
\begin{proof}
We divide the proof into three cases, according to the value of $i$.

The first case is $i=1$. Clearly, if $x$ is a leaf vertex, then $a_{x,1}=a_{x,e}$ and $b_{x,1}=b_{x,e}$. Moreover, $c_{x,1}(a_{x,1})=c_{x,e}(a_{x,e})$ as well as $c_{x,1}(b_{x,1})=c_{x,e}(b_{x,e})$. Therefore, we can assume that $x$ is not a leaf vertex. Let $x_1,\dots,x_{\ell-1}$ be the $\ell-1$ children of $x$ in $T$.  Since $Z(x,1)=\{x\} \cup \bigcup_{i=1}^{\ell-1}Z(x_i,1)$, for every $y \in V(T)$, we have that 
\begin{equation*}
c_{x,1}(y)=\max\Big\{c_{x,e}(y),1+\max_{i=1,\dots,\ell-1}c_{x_i,1}(y)\Big\}.
\end{equation*}
Therefore, using Lemma \ref{lm:merge_diameter_lemma}, for every $x \in X$, all the vertices $a_{x,1},b_{x,1}$, together with their corresponding costs w.r.t. $c_{x,i}$, can be computed in $O(n)$ time and space via a postorder visit of the vertices in $X$. 

We consider the case in which $i=2$ and we assume that, for every $x \in X$, all the vertices $a_{x,1},b_{x,1}$ and their corresponding costs w.r.t. $c_{x,1}$ are known. Let $\bar x$ be the parent of $x$ in $T$ and let $x_1,\dots,x_\ell$ be the $\ell \geq 1$ children of $\bar x$ in $T$. For every $i=1,\dots,\ell$, let $\bar c_{\bar x,i}$ and $\hat c_{\bar x,i}$ be two cost functions that, for every $y \in V(T)$, are defined as follows:
\begin{align*}
\bar c_{\bar x,i}(y) :=2+\max_{j=1,\dots,i-1}c_{x_j,1}(y),\\
\hat c_{\bar x,i}(y) :=2+\max_{j=i+1,\dots,\ell}c_{x_j,1}(y).
\end{align*}
For every $i=2,\dots,\ell-1$, the algorithm computes the vertices $\bar a_{\bar x,i},\bar b_{\bar x,i},\hat a_{\bar x,i},\hat b_{\bar x,i}$ and the costs $\bar c_{\bar x,i}(\bar a_{\bar x,i}),\bar c_{\bar x,i}(\bar b_{\bar x,i}),\hat c_{\bar x,i}(\hat a_{\bar x,i}),\hat c_{\bar x,i}(\hat b_{\bar x,i})$ using simple dynamic programming and Lemma~\ref{lm:merge_diameter_lemma}. Indeed, $\bar c_{\bar x,2}(y)=2+c_{x_1,1}(y)$ as well as $\hat c_{\bar x,\ell-1}(y)=2+c_{x_{\ell},1}(y)$. Furthermore,  for every $i > 2$,
\begin{equation*}
\bar c_{\bar x,i}(y) =\max\big\{\bar c_{\bar x,i-1}(y), 2+c_{x_{i-1},1}(y)\big\},
\end{equation*}
while for every $i < \ell-1$,
\begin{equation*}
\hat c_{\bar x,i}(y) =\max\big\{\hat c_{\bar x, i+1}(y), 2+c_{x_{i+1},1}(y)\big\}.
\end{equation*}
We observe that  $\bar a_{\bar x,i},\bar b_{\bar x,i},\hat a_{\bar x,i},\hat b_{\bar x,i}$ and the costs $\bar c_{\bar x,i}(\bar a_{\bar x,i}),\bar c_{\bar x,i}(\bar b_{\bar x,i}),\hat c_{\bar x,i}(\hat a_{\bar x,i}),\hat c_{\bar x,i}(\hat b_{\bar x,i})$ can be computed in $O(\ell)$ time and space. As a consequence, these pieces of information can be precomputed for every $x \in V(T)$ in $O(n^2)$ time and space.
Let $x=x_i$, for some $i=1,\dots,\ell$.
Since 
\begin{equation*}
Z(x,2)=\bigcup_{j=1,\dots,\ell: j\neq i}Z(x_j,1)=\bigcup_{j=1}^{i-1}Z(x_j,1)\cup \bigcup_{j=i+1}^{\ell}Z(x_j,1),
\end{equation*}
for every $y \in V(T)$, we have that
\begin{equation*}
c_{x,2}(y)=\max\big\{\bar c_{\bar x,i}(y),\hat c_{\bar x,i}(y)\big\}.
\end{equation*}
Therefore, using Lemma \ref{lm:merge_diameter_lemma}, all the vertices $a_{x,2},b_{x,2}$ and their corresponding costs w.r.t. $c_{x,2}$ can be computed in $O(n)$ time and space for every $x \in X$.

Finally, we consider the case in which $i=3$ and we assume that all the vertices $a_{x,j},b_{x,j}$ and the costs $c_{x,j}(a_{x,j}),c_{x,j}(b_{x,j})$, with $x \in X$ and $j=1,2$, are known. If $x=v$, then $Z(x,3)=\emptyset$. Therefore, we only need to prove the claim when $x\neq v$. Let $\bar x$ be the parent of $x$ in $T$.
Since
\begin{equation*}
Z(x,3)=\big\{\bar x\big\}\cup Z(\bar x,2) \cup Z(\bar x,3)
\end{equation*}
for every $y \in V(T)$, we have that 
\begin{equation*}
c_{x,3}(y)=1+\max\big\{c_{\bar x,e}(y),c_{\bar x,2}(y),c_{\bar x,3}(y)\big\}.
\end{equation*}
Therefore, using Lemma \ref{lm:merge_diameter_lemma}, for every $x \in X$, all the vertices $a_{x,3},b_{x,3}$, and the corresponding costs w.r.t. $c_{x,3}$, can be computed in $O(n)$ time and space by a preorder visit of the tree vertices. This completes the proof.
\end{proof}
We can now prove the following.
\begin{lemma}\label{lm:endvertices_diametral_path}
For every $x \in X$, all the vertices $a_x, b_x, \gamma_x$ and the costs $c_x(a_x)$ and $c_x(b_x)$ can be computed in $O(n)$ time and space.
\end{lemma}
\begin{proof}
Let $x \in X$ be fixed. Since $X=Z(x,1)\cup Z(x,2)\cup Z(x,3)$, by definition of $c_{x,i}$, for every $y \in V(T)$, we have that
$c_x(y)=\max\big\{c_{x,1}(y),c_{x,2}(y),c_{x,3}(y)\big\}$.
Therefore, under the assumption that all the vertices $a_{x,i},b_{x,i}$ and all the values $c_{x,i}(a_{x,i})$, $c_{x,i}(b_{x,i})$, with $x \in X$ and $i=1,2,3$, are known, using Lemma \ref{lm:merge_diameter_lemma}, we can compute $a_x$ and $b_x$, together with the values $c_x(a_x)$ and $c_x(b_x)$ in constant time. The claim follows since, as we already discussed at the end of Section~\ref{subsec:how_to_compute_the_candidate_best_swap_edge}, $\gamma_x$ can be computed in constant time.
\end{proof}

\subsection{How to compute the vertex \texorpdfstring{$y_x$}{closest to the center}}
In this subsection we show how to build, in $O(n)$ time and space, a static data structure for each vertex $x$ that, for every $e \in E(T)$, is able to compute the vertex $y_x$ of $Y(x,e)$ closest to $\gamma_x$ in constant time. Let $x \in V(T)$ be fixed and let $r=z_1,\dots, z_h$ be the $h$ proper ancestors of $x$, in the order in which they are encountered by a traversal of the path in $T$ from $r$ to $x$. For every $i=1,\dots,h$, let $V(x|z_i):=\big\{y \in V(T) \mid z_i=\lca(x,y)\big\}$. Observe that $Y(x|z_i) \subseteq V(x|z_i)$. We label each vertex of $V(x|z_i)$ with the vertex of $Y(x|z_i)$ that is closest to it (ties are broken arbitrarily). More precisely, if $y \in Y(x|z_i)$ is the vertex closest to $y' \in V(x|z_i)$, then $\lambda(y')=y$ is the label of $y'$. We also build two {\em range-minimum-query} data structures $\mathcal{R}$ and $\mathcal{R}'$, i.e., two vectors of $h$ elements each such that, for each $i=1,\dots,h$, $\mathcal{R}[i]=d_T\big(z_i,\lambda(z_i)\big)+i$ and $\mathcal{R}'[i]=d_T\big(z_i,\lambda(z_i)\big)-i$. The following lemma holds.
\begin{lemma}\label{lm:data_structures}
The labeling $\lambda$ and the two range-minimum-query data structures $\mathcal{R}$ and $\mathcal{R}'$ can be computed in $O(n)$ time and space.
\end{lemma}
\begin{proof}
It is known that the range-minimum-query data structure of size $h$ can be computed in $O(h)$ time and space \cite{DBLP:journals/tcs/BenderF04, DBLP:conf/stoc/GabowT83}. The labeling $\lambda$ can be computed in $O(n)$ time and space by a simple algorithm that, for every $i=1,\dots,h$, first initializes $\lambda(y)=y$ for every $y \in Y(x,z_i)$ and then, during phase $\phi$, labels all the still unlabeled vertices which are at distance $\phi$ from some vertex in $Y(x|z_i)$.
\end{proof}
Let $e=(z_i=\parent(v),v)$ be the failing edge. We show how to find the vertex $y_x \in Y(x,e)$ that is closest to $\gamma_x$ in constant time. First we compute $j$ such that $z_j=\lca(\gamma_x,x)$. Next we make at most two range minimum queries to compute the following two indices:
\begin{itemize}
\item the index $t'$ containing the minimum value within the range $[1,j-1]$ in $\mathcal{R}'$;
\item the index $t$ containing the minimum value within the range $[j+1,i]$ in $\mathcal{R}$.
\end{itemize}
The algorithm chooses $y_x$ such that
\begin{equation*}
y_x \in \arg\min_{y \in \{\lambda(\gamma_x),\lambda(z_{t'}),\lambda(z_t)\}}d_T(y,\gamma_x),
\end{equation*}

\begin{lemma}\label{lm:vertex_closest_to_center}
The vertex $y_x$ selected by the algorithm satisfies $y_x \in \arg\min_{y \in Y(x,e)}d_T(y,\gamma_x)$.
\end{lemma}
\begin{proof}
Let $y^* \in \arg\min_{y \in Y(x,e)}d_T(y,\gamma_x)$. Clearly, for some $k=1,\dots,i$, $y^* \in Y(x|z_k)$. We prove the claim by showing that $d_T(y_x,\gamma_x)\leq d_T(y^*,\gamma_x)$. We divide the proof into three cases, according to the value of $k$.

\noindent The first case is when $k=j$. We have that $d_T(y_x,\gamma_x) \leq d_T\big(\lambda(\gamma_x),\gamma_x\big) = d_T(y^*,\gamma_x)$.

The second case occurs when $k < j$. Clearly, $d_T\big(\lambda(z_k),z_k\big) \leq d_T(y^*,z_k)$. Moreover, $d_T\big(\lambda(z_{t'}),z_{t'}\big)-t' \leq d_T\big(\lambda(z_k),z_k\big)-k$.
Therefore,
\begin{align*}
d_T(y_x,\gamma_x) 	& \leq d_T\big(\lambda(z_{t'}),\gamma_x\big) = d_T\big(\lambda(z_{t'}),z_{t'}\big)+d_T(z_{t'},z_j)+d_T(z_j,\gamma_x)\\
				& = d_T\big(\lambda(z_{t'}),z_{t'}\big) + j - t' + d_T(z_j,\gamma_x) \leq d_T\big(\lambda(z_k),z_k\big) + j - k + d_T(z_j,\gamma_x)\\
				& \leq d_T(y^*,z_k)+j-k+d_T(z_j,\gamma_x) = d_T(y^*,\gamma_x).
\end{align*}

The third case occurs when $j < k \leq i$. Clearly, $d_T\big(\lambda(z_k),z_k\big) \leq d_T(y^*,z_k)$. Moreover, $d_T\big(\lambda(z_t),z_t\big)+t \leq d_T\big(\lambda(z_k),z_k\big)+k$. Therefore,
\begin{align*}
d_T(y_x,\gamma_x) 	& \leq d_T\big(\lambda(z_t),\gamma_x\big) = d_T\big(\lambda(z_t),z_t\big)+d_T(z_t,z_j)+d_T(z_j,\gamma_x)\\
				& = d_T\big(\lambda(z_t),z_t\big) + t-j + d_T(z_j,\gamma_x) \leq d_T\big(\lambda(z_k),z_k\big) + k - j + d_T(z_j,\gamma_x)\\
				& \leq d_T(y^*,z_k)+k-j+d_T(z_j,\gamma_x) = d_T(y^*,\gamma_x).
\end{align*}
The claim follows.
\end{proof}
We can finally state the main theorem.
\begin{theorem}
All the best swap edges of a tree spanner $T$ in 2-edge-connecte, unweighted, and undirected graphs can be computed in $O(n^2)$ time and space.
\end{theorem}
\begin{proof}
From~Lemma \ref{lm:endvertices_diametral_path}, for a fixed edge $e \in E(T)$, all the vertices $a_x, b_x, \gamma_x$ and all the values $c_x(a_x),c_x(b_x)$, with $x \in X$, can be computed in $O(n)$ time and space. Therefore, such vertices and values can be computed for every edge of $T$ in $O(n^2)$ time and space.

By Lemma \ref{lm:vertex_closest_to_center}, for a fixed edge $e$ of $T$ and a fixed vertex $x$, we can compute $y_x$, i.e., $f_x=(x,y_x)$, by making at most two queries, each of which requires constant time, on the two range-minimum-query data structures associated with $x$. Therefore, the $O(n)$ candidate best swap edges of $e$ can be computed in $O(n)$ time. Furthermore, using Lemma \ref{lm:sirocco}, we can compute $\sigma\big(T_{e/f_x}\big)=1+\max\big\{d_T(y_x,a_x)+c_x(a_x),d_T(y_x,b_x)+c_x(b_x)\big\}$ in constant time. Hence, thanks to Lemma~\ref{lm:best_swap_edge_computation}, the best swap edge $f^*$ of $e$ can be computed in $O(n)$ time. The claim follows.
\end{proof}

\bibliography{bibliography}

\end{document}